\documentclass[draft,10pt,a4paper]{amsart}

\copyrightinfo{2021}{Nicholas Coxon}

\calclayout

\usepackage{amsaddr}

\usepackage{etoolbox}
\makeatletter
\patchcmd{\abstract}{3pc}{25mm}{}{}
\makeatother

\newtheorem{theorem}{Theorem}[section]
\newtheorem{lemma}[theorem]{Lemma}
\theoremstyle{definition}
\newtheorem{algorithm}[theorem]{Algorithm}

\numberwithin{equation}{section}

\usepackage[noend]{algpseudocode}
\usepackage{algorithm}

\renewcommand{\algorithmicensure}{\textbf{Result:}}
\renewcommand{\textproc}{\textsf}

\def\vec#1{\mathchoice{\mbox{\boldmath$\displaystyle#1$}}
{\mbox{\boldmath$\textstyle#1$}}
{\mbox{\boldmath$\scriptstyle#1$}}
{\mbox{\boldmath$\scriptscriptstyle#1$}}}

\def\ceil#1{\mathchoice{\mbox{$\displaystyle\left\lceil#1\right\rceil$}}
{\mbox{$\textstyle\lceil#1\rceil$}}
{\mbox{$\scriptstyle\lceil#1\rceil$}}
{\mbox{$\scriptscriptstyle\lceil#1\rceil$}}}

\def\floor#1{\mathchoice{\mbox{$\displaystyle\left\lfloor#1\right\rfloor$}}
{\mbox{$\textstyle\lfloor#1\rfloor$}}
{\mbox{$\scriptstyle\lfloor#1\rfloor$}}
{\mbox{$\scriptscriptstyle\lfloor#1\rfloor$}}}

\newcommand{\Z}{\mathbb{Z}}
\newcommand{\N}{\mathbb{N}}
\newcommand{\C}{\mathbb{C}}
\newcommand{\A}{\mathcal{R}}
\newcommand{\bigO}{O}
\newcommand{\set}{\leftarrow}
\newcommand{\x}[1]{x_{#1}}
\newcommand{\concat}{\mathbin{\Vert}}
\DeclareMathOperator{\DFT}{DFT}
\DeclareMathOperator{\TFT}{TFT}
\DeclareMathOperator{\ord}{ord}
\DeclareMathOperator{\Div}{div}

\begin{document}
\title{An in-place truncated Fourier transform}
\author{Nicholas Coxon}
\email{coxon.nv@gmail.com}
\date{\today}

\begin{abstract}
We show that simple modifications to van der Hoeven's forward and inverse truncated Fourier transforms allow the algorithms to be performed in-place, and with only a linear overhead in complexity.
\end{abstract}

\maketitle

\section{Introduction}

The discrete Fourier transform (DFT) with respect to a principal $n$th root of unity is a linear map that evaluates polynomials of degree less than $n$ at each of the~$n$ distinct powers of the root of unity. A naive approach allows the map to be evaluated with $\bigO(n^2)$ arithmetic operations. However, when $n$ is comprised entirely of small prime factors, the DFT can be evaluated with $\bigO(n\log n)$ operations by the fast Fourier transform (FFT) algorithm. While already known to Gauss~\cite{gauss1866}, the FFT algorithm only received widespread attention after its rediscovery~\cite{heideman1985,cooley1987} by Cooley and Tukey~\cite{cooley1965}. The algorithm has since found numerous applications in mathematics, computer science and electrical engineering. In the area of computer algebra, the FFT algorithm provides the basis for asymptotically fast algorithms for integer and polynomial multiplication~\cite{schonage1971,cantor1991}.

The radix-2 Cooley-Tukey FFT algorithm is widely used due to its ease of implementation and fast practical performance. However, the radix-2 algorithm requires the order $n$ of the root of unity to be a power of two, limiting control over the number of evaluation points of the corresponding transform. This requirement leads to unwanted ``jumps'' in the complexity for certain applications, e.g., polynomial multiplication, since polynomials may have to be evaluated at almost twice as many points than actually required by the problem at hand. Truncated (or pruned) variants of the radix-2 algorithm address this issue by considering more general problems involving the evaluation of polynomials with restricted support at subsets of the evaluation points. The algorithms then disregard parts of the radix-2 algorithm that are vacuous due to zero coefficients, or that do not contribute to the computation of the desired output, resulting in complexities which vary relatively smoothly with the size of the inputs and outputs.

The truncated Fourier transform (TFT) algorithm of van der Hoeven~\cite{hoeven2004,hoeven2005} modifies the radix-2 algorithm so that it evaluates polynomials of degree less than $\ell$, where $\ell\leq n$, at the $\ell$ powers of the root of unity whose exponents are smallest under a reversal of their bit representation. The algorithm retains the simplicity of the radix-2 algorithm, while enjoying a time complexity which varies relatively smoothly with~$\ell$. More specifically, the algorithm performs $(\ell/2)\log_2\ell+\bigO(\ell)$ ring multiplications, and $\ell\log_2\ell+\bigO(\ell)$ ring additions or subtractions. Van der Hoeven also provides a second algorithm, called the inverse truncated Fourier transform~(ITFT), which performs the inverse transformation of interpolation with the same complexity. Together, the algorithms allow for more efficient implementations of FFT-based polynomial multiplication algorithms by reducing the size of the jumps in their complexities.

Van der Hoeven's TFT and ITFT algorithms operate on arrays of $2^{\ceil{\log_2\ell}}$ ring elements, resulting in jumps in their space complexities as $\ell$ increases. Harvey and Roche~\cite{harvey2010} address this issue by describing in-place variants of the algorithms, which operate by updating their length $\ell$ input arrays through replacements until they are overwritten with the output, while storing only a constant number of ring elements and bounded-precision integers in auxiliary storage space. Their variants retain the $\bigO(\ell\log\ell)$ time complexity of van der Hoeven's algorithms, albeit with a larger implied constant. Arnold~\cite{arnold2013} subsequently describes in-place algorithms that match the complexity of van der Hoeven's algorithms up to the implied constants in the $\bigO(\ell)$ terms. The algorithms are based on existing in-place algorithms~\cite{sergeev2011} for computing a different truncated transform~\cite[Chapter~6]{mateer2008}, which Arnold improves and then modifies to compute the transforms of van der Hoeven's algorithms.

In Section~\ref{sec:tft} of this paper, we present an in-place TFT algorithm that performs at most $(\ell/2)\floor{\log_2\ell}+2\ell+\bigO(\log^2\ell)$ ring multiplications, and at most $\ell\floor{\log_2\ell}+2\ell$ ring additions or subtractions. The algorithm simply injects an additional $\bigO(\ell)$ computation into the middle of a slightly restructured version of van der Hoeven's TFT algorithm. This approach is arguably more natural than the one used by Arnold, which we see reflected in the fact that our algorithm performs $\Omega(\ell)$ fewer multiplications than Arnold's algorithm when $\ell$ is not a power of two. In Section~\ref{sec:itft}, we show that our TFT algorithm is easily inverted, yielding an in-place ITFT algorithm that performs at most $(\ell/2)\floor{\log_2\ell}+4\ell+\bigO(\log^2\ell)$ ring multiplications, and at most $\ell\floor{\log_2\ell}+3\ell$ ring additions or subtractions.

\section{Preliminaries}

\subsection{The discrete Fourier transform}

Let $\A$ be a commutative ring with identity. Then $\omega\in\A$ is called a principal $n$th root of unity if $\omega^n=1$ and $\sum^{n-1}_{j=0}\omega^{ij}=0$ for $i\in\{1,\dotsc,n-1\}$. The discrete Fourier transform (DFT) with respect to a principal $n$th root of unity $\omega\in\A$ is the $\A$-linear map $\DFT_\omega:\A^n\rightarrow\A^n$ defined by $(a_0,\dotsc,a_{n-1})\mapsto(\hat{a}_0,\dotsc,\hat{a}_{n-1})$ where
\begin{equation*}
\hat{a}_i
=
\sum^{n-1}_{j=0}
a_j
\omega^{ij}
\quad\text{for $i\in\{0,\dotsc,n-1\}$}.
\end{equation*}
It is readily verified that $\omega^{-1}=\omega^{n-1}$ is also a principal $n$th root of unity, and $\DFT_{\omega^{-1}}(\DFT_\omega(\vec{a}))=n\vec{a}$ for $\vec{a}\in\A^n$. Thus, $\DFT_\omega$ is injective if and only if $n$ is not a zero-divisor in $\A$.

Hereafter, we assume that $\omega\in\A$ is a principal $n$th root of unity such that $n=2^p$ for some $p\geq 1$. We further assume that $\omega^{n/2}=-1$, in order to simplify the algorithms presented in this paper. When $n$ is a power of two, satisfying this condition is sufficient for $\omega\in\A$ to be an $n$th principal root of unity; the condition is also necessary if $2$ is not a zero-divisor in $\A$.

\subsection{Computational model}\label{sec:computational_model}

We adhere to the computational model laid out by Harvey and Roche~\cite[Section~2.1]{harvey2010}. Accordingly, two primitive types are allowed in memory: elements of $\A$ and integers in the interval $[-c\ell,c\ell]$, where $\ell$ is the length of the transform being computed, and $c\in\N$ is a fixed constant (like Harvey and Roche, we may take $c=2$ for our algorithms). An algorithm is then said to be in-place if it operates by updating its given inputs until they are overwritten with the corresponding outputs. The algorithm may use extra working space during the computation, which is defined to be the auxiliary space used by the algorithm. However, only $\bigO(1)$ auxiliary space is allowed.

We measure the time complexity of algorithms by the number algebraic operations they perform, i.e., by the number of multiplications, additions and subtractions performed in $\A$. Thus, we ignore the cost of operations on indices. However, these costs are negligible compared to those of the algebraic operations for the algorithms in this paper. At times, we provide a more detailed analysis by separately counting the number of multiplications by powers of $\omega$, and by powers of $2$ or its inverse. Multiplications by powers of $\omega$ are significant in settings where the root has been artificially added to the ring~\cite{schonage1971,cantor1991}, since they enjoy a smaller cost than arbitrary multiplications. Similarly, the cost of a multiplication by $2$ is always bounded by that of an addition, while for certain rings (e.g., $\A=\C$) multiplications by powers of $2$ or its inverse can be efficiently implemented by simple bit operations (e.g., bit-shifts).

\subsection{The radix-2 fast Fourier transform}\label{sec:FFT}

Along with the rediscovery of the FFT algorithm, Cooley and Tukey~\cite{cooley1965} introduced an in-place variant of the radix-2 algorithm. A feature of the algorithm is that it requires a ``bit-reversal'' permutation to be performed either at its beginning or at its end. For $k\in\N$, the $k$-bit bit-reversal permutation $[\,{}\,]_k:\{0,\dotsc,2^k-1\}\rightarrow\{0,\dotsc,2^k-1\}$ is defined by
\begin{equation*}
	2^0i_0+2^1i_1+\dotsb+2^{k-1}i_{k-1}
	\mapsto
	2^0i_{k-1}+2^1i_{k-2}+\dotsb+2^{k-1}i_0
\end{equation*}
for $i_0,\dotsc,i_{k-1}\in\{0,1\}$. The in-place radix-2 FFT algorithm is then derived from the following result of Cooley and Tukey (see also~\cite[Section~2]{hoeven2004}).

\begin{lemma}\label{lem:radix-2_FFT} Let $\vec{a}=(a_0,\dotsc,a_{n-1})\in\A^n$, and recursively define $a_{k,0},\dotsc,a_{k,n-1}\in\A$ for $k\in\{0,\dotsc,p\}$ as follows: $a_{p,i}=a_i$ for $i\in\{0,\dotsc,n-1\}$; and, if $k<p$,
\begin{equation}\label{eqn:butterfly}
	\begin{pmatrix}
		a_{k,2^k(2i  )+j} \\
		a_{k,2^k(2i+1)+j}
	\end{pmatrix}
	=
	\begin{pmatrix}
		1 &  \omega^{[2i]_p} \\
		1 & -\omega^{[2i]_p}
	\end{pmatrix}
	\begin{pmatrix}
		a_{k+1,2^k(2i  )+j} \\
		a_{k+1,2^k(2i+1)+j}
	\end{pmatrix}
\end{equation}
for $i\in\{0,\dotsc,2^{p-k-1}-1\}$ and $j\in\{0,\dotsc,2^k-1\}$. Then $\DFT_\omega(\vec{a})=(a_{0,[0]_p},a_{0,[1]_p},\dotsc,a_{0,[n-1]_p})$.
\end{lemma}

%

Lemma~\ref{lem:radix-2_FFT} immediately yields a $\bigO(n\log n)$ algorithm for evaluating $\DFT_\omega$. To make the algorithm in-place, the input vector $(a_0,\dotsc,a_{n-1})$ is successively overwritten with $a_{k,0},\dotsc,a_{k,n-1}$ for $k=p-1,p-2,\dotsc,0$ by using~\eqref{eqn:butterfly} to overwrite pairs of entries. The vector is then permuted in-place to obtain the correct ordering, which simply involves swapping $n/2$ pairs of entries since the bit-reversal permutation is an involution. It is necessary to compute the multipliers $\omega^{[2i]_p}$ of the transformations~\eqref{eqn:butterfly}, commonly referred to as ``twiddle factors'', using only $\bigO(1)$ auxiliary space. This requirement can be met by performing the transformations  successively for $i=[0]_{p-k-1},[1]_{p-k-1},\dotsc,[2^{p-k-1}-1]_{p-k-1}$, since the corresponding twiddle factors are then successive powers of~$\omega^{2^k}$. More simply, the input vector can be permuted according to the bit-reversal permutation at the beginning of the algorithm, rather than at its end, since~\eqref{eqn:butterfly} implies that
\begin{equation*}
	\begin{pmatrix}
		a_{k,[2^{p-k-1}(2j  )+i]_p} \\
		a_{k,[2^{p-k-1}(2j+1)+i]_p}
	\end{pmatrix}
	=
	\begin{pmatrix}
		1 &  \omega^{2^ki} \\
		1 & -\omega^{2^ki}
	\end{pmatrix}
	\begin{pmatrix}
		a_{k+1,[2^{p-k-1}(2j  )+i]_p} \\
		a_{k+1,[2^{p-k-1}(2j+1)+i]_p}
	\end{pmatrix}
\end{equation*}
for $k\in\{0,\dotsc,p-1\}$, $i\in\{0,\dotsc,2^{p-k-1}-1\}$ and $j\in\{0,\dotsc,2^k-1\}$. Using either method, the cost of computing the twiddle factors is $n+\bigO(\log^2n)$ multiplications by powers of $\omega$.

\begin{theorem}[\cite{cooley1965}]\label{thm:radix-2_FFT} If $\vec{a}\in\A^n$, then $\DFT_\omega(\vec{a})$ can be computed in-place with $(n/2)\log_2 n+\bigO(\log^2n)$ multiplications by powers of $\omega$, and $n\log_2n$ additions or subtractions.
\end{theorem}

\subsection{The truncated Fourier transform}

The length $\ell\leq n$ truncated Fourier transform (TFT) with respect to $\omega$ is the $\A$-linear map $\TFT_{\omega,\ell}:\A^\ell\rightarrow\A^\ell$ defined by $(a_0,\dotsc,a_{\ell-1})\mapsto (\hat{a}_{[0]_p},\dotsc,\hat{a}_{[\ell-1]_p})$ where
\begin{equation*}
	\hat{a}_i
	=
	\sum^{\ell-1}_{j=0}
	a_j
	\omega^{ij}
	\quad\text{for $i\in\{0,\dotsc,n-1\}$}.
\end{equation*}
When convenient, we also view the truncated Fourier transform as the polynomial evaluation map $\TFT_{\omega,\ell}:\A[x]_{<\ell}\rightarrow\A^\ell$ given by $f\mapsto(f(\omega^{[0]_p}),\dotsc,f(\omega^{[\ell-1]_p}))$, where $\A[x]_{<\ell}$ denotes the set of polynomials over $\A$ with degree less than $\ell$.

The length $\ell$ TFT is readily evaluated by applying Lemma~\ref{lem:radix-2_FFT} with $a_\ell=\dotsb=a_{n-1}=0$. Van der Hoeven's TFT algorithm simply augments this approach by disregarding redundant parts of the computation. For example, the redundancies introduced by the zero padding are accounted for by the following lemma. 

\begin{lemma}\label{lem:level_m-1} Let $\ell\in\{2,\dotsc,n\}$, $m=\ceil{\log_2\ell}$ and $a_0,\dotsc,a_{n-1}\in\A$ such that $a_\ell=\dotsb=a_{n-1}=0$. Define $a_{k,0},\dotsc,a_{k,n-1}\in\A$ for $k\in\{0,\dotsc,p\}$ as per Lemma~\ref{lem:radix-2_FFT}. Then
\begin{equation}\label{eqn:level_m-1}
	a_{m-1,2^{m-1}i+j}
	=
	\begin{cases}
		a_j+(-1)^ia_{2^{m-1}+j} & \text{if $j<\ell-2^{m-1}$},    \\
		a_j                     & \text{if $j\geq\ell-2^{m-1}$},
	\end{cases}
\end{equation}
for $i\in\{0,1\}$ and $j\in\{0,\dotsc,2^{m-1}-1\}$.
\end{lemma}
\begin{proof} If for some $k\in\{m,\dotsc,p-1\}$, we have $a_{k+1,j}=a_j$ for $j\in\{0,\dotsc,2^{k+1}-1\}$, then \eqref{eqn:butterfly} with $i=0$ implies that $a_{k,j}=a_{k+1,j}+a_{k+1,2^k+j}=a_j+a_{2^k+j}=a_j$ for $j\in\{0,\dotsc,2^k-1\}$, since $2^k\geq 2^m\geq\ell$. As $a_{p,j}=a_j$ for $j\in\{0,\dotsc,2^p-1\}$ by definition, it follows by induction that $a_{m,j}=a_j$ for $j\in\{0,\dotsc,2^m-1\}$. Thus,~\eqref{eqn:butterfly} with $k=m-1$ and $i=0$ implies that $a_{m-1,j}=a_j+a_{2^{m-1}+j}$ and $a_{m-1,2^{m-1}+j}=a_j-a_{2^{m-1}+j}$ for $j\in\{0,\dotsc,2^{m-1}-1\}$, where $a_{2^{m-1}+j}=0$ if $j\geq\ell-2^{m-1}$.
\end{proof}

To compute the truncated Fourier transform of $(a_0,\dotsc,a_{\ell-1})\in\A^\ell$, van der Hoeven's algorithm operates on an array of length $2^{\ceil{\log_2\ell}}$ that initially has its first $\ell$ entries set equal to $a_0,\dotsc,a_{\ell-1}$. The algorithm then overwrites the first $2^k\ceil{\ell/2^k}$ entries of the array with $a_{k,0},\dotsc,a_{k,2^k\ceil{\ell/2^k}-1}$, successively for $k=\ceil{\log_2\ell}-1,\ceil{\log_2\ell}-2,\dotsc,0$, by using \eqref{eqn:level_m-1} if $k=\ceil{\log_2\ell}-1$, and otherwise using~\eqref{eqn:butterfly} to overwrite pairs of entries or single entries. The truncated Fourier transform is then read from the first $\ell$ entries of the array upon termination.

The description of van der Hoeven's ITFT algorithm is more involved, and not provided here. However, we do recall a simple, but key, observation that is used in its development. Van der Hoeven~\cite[Section~4]{hoeven2004} (see also~\cite[Section~6]{hoeven2005}) observes that not only do we have~\eqref{eqn:butterfly}, but also the relations
\begin{equation}\label{eqn:inverse_butterfly}
	\begin{pmatrix}
		2a_{k+1,2^k(2i  )+j} \\
		2a_{k+1,2^k(2i+1)+j}
	\end{pmatrix}
	=
	\begin{pmatrix}
		1 &  1 \\
		\omega^{-[2i]_p} & -\omega^{-[2i]_p}
	\end{pmatrix}
	\begin{pmatrix}
		a_{k,2^k(2i  )+j} \\
		a_{k,2^k(2i+1)+j}
	\end{pmatrix}
\end{equation}
and
\begin{equation}\label{eqn:vdH}
	\begin{pmatrix}
		a_{k+1,2^k(2i)+j}\\
		a_{k,2^k(2i+1)+j}
	\end{pmatrix}
	=
	\begin{pmatrix}
		1 & -\omega^{[2i]_p}\\
		1 & -2\omega^{[2i]_p}
	\end{pmatrix}
	\begin{pmatrix}
		a_{k,2^k(2i)+j}\\
		a_{k+1,2^k(2i+1)+j}
	\end{pmatrix}
\end{equation}
for $k\in\{0,\dotsc,p-1\}$, $i\in\{0,\dotsc,2^{p-k-1}-1\}$ and $j\in\{0,\dotsc,2^k-1\}$, and those obtained by inverting the matrix in~\eqref{eqn:vdH}. The additional relations are used in van der Hoeven's ITFT algorithm to navigate nimbly, and cheaply, through the elements $a_{k,i}$ for increasing \emph{and} decreasing~$k$. They are used similarly in our in-place algorithms. 

The complexity of van der Hoeven's TFT and ITFT algorithms is summarised as follows.

\begin{theorem}[{\cite[Theorems~1 and~2]{hoeven2004}}]\label{thm:vdH} Suppose that $\vec{a}\in\A^\ell$ for some $\ell\leq n$. Then $\TFT_{\omega,\ell}(\vec{a})$ can be computed with $(\ell/2)\log_2\ell+\bigO(\ell)$ multiplications by powers of $\omega$, and $\ell\log_2\ell+\bigO(\ell)$ additions or subtractions. Moreover, if $2$ is a unit in $\A$, then $\vec{a}$ can be recovered from $\TFT_{\omega,\ell}(\vec{a})$ with $(\ell/2)\log_2\ell+\bigO(\ell)$ multiplications by powers of $\omega$, and $\ell\log_2\ell+\bigO(\ell)$ shifted additions or subtractions.
\end{theorem}

In Theorem~\ref{thm:vdH}, a ``shifted'' addition or subtraction refers to the combined operation of an addition or subtraction with a multiplication by $2$ or $2^{-1}$. These operations arise, for example, when using~\eqref{eqn:inverse_butterfly} and~\eqref{eqn:vdH}, and are significant for the reasons discussed in Section~\ref{sec:computational_model}.

\subsection{In-place algorithms for the truncated Fourier transform}

Van der Hoeven's TFT and ITFT algorithms both operate on arrays of $2^{\ceil{\log_2\ell}}$ ring elements. Thus, their auxiliary space requirements are not $\bigO(1)$, since each algorithms stores at least $2^{\ceil{\log_2\ell}}-\ell$ ring elements auxiliary space. The first in-place algorithms for the truncated Fourier transform are due to Harvey and Roche~\cite{harvey2010} (see also \cite[Chapter~3]{roche2011}). Their TFT algorithm is again based on the idea of computing the sequences $a_{k,0},\dotsc,a_{k,2^k\ceil{\ell/2^k}-1}$. However, elements $a_{k,i}$ with $i>\ell$ are never stored by the algorithm. Instead, they are computed on-the-fly only when needed, and then immediately discarded. This approach allows the algorithm to operate on a vector of length $\ell$, and potentially results in the computation of fewer elements $a_{k,i}$ with $i>\ell$. However, the individual cost of their computation is higher, leading to an increase in multiplicative complexity over van der Hoeven's algorithm of a constant factor.

\begin{theorem}[{\cite[Chapter~3]{roche2011}}] If $\vec{a}\in\A^\ell$ for some $\ell\leq n$, then $\TFT_{\omega,\ell}(\vec{a})$ can be computed in-place with $5\ell(\log_2\ell)/6+\bigO(\ell)$ multiplications by powers of $\omega$, and $\bigO(\ell\log\ell)$ additions or subtractions.
\end{theorem}

Harvey and Roche's in-place ITFT algorithm performs a similar number of multiplications by powers of~$\omega$, and $\bigO(\ell\log\ell)$ shifted additions or subtractions.

Arnold~\cite{arnold2013} provides in-place TFT and ITFT algorithms that match the complexity of van der Hoeven's algorithms up to lower order terms. The problem of evaluating $\TFT_{\omega,\ell}$ is not considered directly by Arnold. Instead, the main focus of Arnold's work is to improve an existing in-place algorithm for evaluating the ``cyclotomic'' truncated Fourier transform proposed by Mateer~\cite[Chapter~6]{mateer2008}. The improved algorithm is then modified to evaluate $\TFT_{\omega,\ell}$. We summarise Arnold's main results relating to the evaluation of $\TFT_{\omega,\ell}$ in the following two lemmas.

\begin{lemma}[{\cite[Section~6]{arnold2013}}]\label{lem:tft_decomposition} Write $\ell\in\{1,\dotsc,n\}$ as $\sum^t_{r=1}\ell_r$ for integer powers of two $\ell_1>\ell_2>\dotsb>\ell_t$, and let $e_r=[\ell_r+\dotsb+\ell_t]_p$ for $r\in\{1,\dotsc,t\}$, with $[n]_p$ taken to be zero. Then, for $f\in\A[x]_{<\ell}$,
\begin{equation*}
	\TFT_{\omega,\ell}(f)
	=
	\TFT_{\omega^{n/\ell_1},\ell_1}\left(f_1(x/\omega^{e_1})\right)
	\concat
	\dotsb
	\concat
	\TFT_{\omega^{n/\ell_t},\ell_t}\left(f_s(x/\omega^{e_t})\right),
\end{equation*}
where $f_r\in\A[x]_{<\ell_r}$ is the residue of $f(\omega^{[\ell]_p}x)$ modulo $x^{\ell_r}+1$, for $r\in\{1,\dotsc,t\}$, and $\concat$ denotes the concatenation operator.
\end{lemma}


\begin{lemma}[{\cite[Lemma~13]{arnold2013}}]\label{lem:in-place_residues} Write $\ell\in\N\setminus\{0\}$ as $\sum^t_{r=1}\ell_t$ for integer powers of two $\ell_1>\ell_2>\dotsb>\ell_t$. Then given $f\in\A[x]_{<\ell}$, the minimum degree residues of $f$ modulo the polynomials $x^{\ell_r}+1$ for $r\in\{1,\dotsc,t\}$ can be computed in-place with at most $2\ell$ multiplications by $2$ or $2^{-1}$, and at most $3\ell$ additions or subtractions.
\end{lemma}

The maps $\TFT_{\omega^{n/\ell_r},\ell_r}$ that appear in Lemma~\ref{lem:tft_decomposition} can be evaluated in-place by the methods of Section~\ref{sec:FFT}, since $\omega^{n/\ell_r}$ is a principal $\ell_r$th root of unity. In Lemma~\ref{lem:in-place_residues}, ``computed in-place'' means that the length $\ell$ coefficient vector of~$f$ is overwritten with the concatenation of the coefficient vectors of the residues, appearing in order, and with the residue modulo $x^{\ell_r}+1$ having coefficient vector of length~$\ell_r$. Arnold furthermore shows that the bounds of the lemma also apply to the cost of inverting this transformation in-place. Thus, by combining Lemmas~\ref{lem:tft_decomposition} and~\ref{lem:in-place_residues} with Theorem~\ref{thm:radix-2_FFT}, Arnold obtains the following complexity result.

\begin{theorem}[\cite{arnold2013}]\label{thm:arnold} Suppose that $2$ is a unit in $\A$, and let $\vec{a}\in\A^\ell$ for some $\ell\leq n$. Then given either~$\vec{a}$ or $\TFT_{\omega,\ell}(\vec{a})$, the other vector can be computed in-place with $(\ell/2)\log_2\ell+\bigO(\ell)$ multiplications, and $\ell\log_2\ell+\bigO(\ell)$ additions or subtractions.
\end{theorem}

The assumption in Theorem~\ref{thm:arnold} that $2$ is a unit in $\A$ can be removed for the case of evaluating $\TFT_{\omega,\ell}$ by replacing the algorithm of Lemma~\ref{lem:in-place_residues} with the earlier algorithm of Sergeev~\cite{sergeev2011}. The algorithms presented in this paper yield an alternative proof of the theorem, and without the need for the unit assumption when evaluating the truncated Fourier transform.

\section{An in-place truncated Fourier transform}\label{sec:tft}

Our in-place algorithm for computing the truncated Fourier transform is presented in Algorithm~\ref{alg:tft}. The algorithm operates on an array $(\x{0},\dotsc,\x{\ell-1})$, which initially contains a given vector $\vec{a}\in\A^\ell$, and has it entries overwritten by the algorithm with those of $\TFT_{\omega,\ell}(\vec{a})$. The algorithm also asks for $\omega^{n/2^{\ceil{\log_2\ell}}}$ as an input, thus assuming (along with Harvey and Roche~\cite[Section~3.3]{roche2011}) a precomputation of $p-\ceil{\log_2\ell}$ multiplications if repeated squaring is used. This element generates the smallest subgroup that contains the twiddle factors used in the algorithm, and is used as part of their computation. Algorithm~\ref{alg:tft} omits the details of the twiddle factor computations, which are instead provided in Section~\ref{sec:twiddle_factors}. We temporarily assume that $2$ is not a zero-divisor in $\A$, so that in Line~\ref{tft:push_up_right_j} of the algorithm the map $\Div_2:2\A\rightarrow\A$ given by $2a\mapsto a$ for $a\in\A$ is well-defined. We also assume that the map can be evaluated in-place. In~Section~\ref{sec:avoiding_inv_2}, we show how to modify the algorithm so as to remove the use of the map, thus removing also the need for the two assumptions.

\begin{algorithm}
	\caption{
		$\textproc{TFT}(
			\ell,
			\omega^{n/2^{\ceil{\log_2\ell}}},
			(\x{0},\dotsc,\x{\ell-1}))
		$}
	\label{alg:tft}
	\begin{algorithmic}[1]
		\Require the transform length $\ell\in\{1,\dotsc,n\}$, the precomputed element $\omega^{n/2^{\ceil{\log_2\ell}}}$, and an array $(\x{0},\dotsc,\x{\ell-1})$ containing some vector~$\vec{a}\in\A^\ell$.
		\Ensure the vector $\TFT_{\omega,\ell}(\vec{a})$ written to $(\x{0},\dotsc,\x{\ell-1})$.
		\State\label{tft:m}$m\set\ceil{\log_2\ell}$, $v\set\ord_2\ell$
		\State%
			\texttt{%
				// Compute $a_{m-1,0},\dotsc,a_{m-1,\ell-1}$
			}
		\For{$j=0,\dotsc,\ell-2^{m-1}-1$}\label{tft:level_m-1}
			\State\label{tft:level_m-1_butterfly}
				$\begin{pmatrix}
					\x{j}         \\
					\x{2^{m-1}+j}
				\end{pmatrix}
				\set
				\begin{pmatrix}
					1 &  1 \\
					1 & -1
				\end{pmatrix}
				\begin{pmatrix}
					\x{j}         \\
					\x{2^{m-1}+j}
				\end{pmatrix}$	
		\EndFor\label{tft:level_m-1_end}
		\State
			\texttt{%
				// Compute $a_{k,2^{k+1}\floor{\ell/2^{k+1}}},\dotsc,a_{k,2^k\ceil{\ell/2^k}-1}$ for $k=m-2,m-3,\dotsc,v$
			}
		\For{$k=m-2,m-3,\dotsc,v$}\label{tft:push_down_right}
			\State
				\label{tft:push_down_right_k}%
				\label{tft:q_and_r}%
				$q\set\floor{\ell/2^{k+1}}$,
				$r\set\ell-2^{k+1}q$,
				$q'\set q-2^{m-k-2}$,
				$\alpha\set\omega^{[2q]_p}$
			\If{$r>2^k$}
				\For{$j=0,\dotsc,r-2^k-1$}\label{tft:push_down_both}
					\State\label{tft:rightmost_butterfly_i}
						$\begin{pmatrix}
							\x{2^k(2q)  +j} \\
							\x{2^k(2q+1)+j}
						\end{pmatrix}
						\set
						\begin{pmatrix}
							1 &  \alpha \\
							1 & -\alpha
						\end{pmatrix}
						\begin{pmatrix}
							\x{2^k(2q)  +j} \\
							\x{2^k(2q+1)+j}
						\end{pmatrix}$
				\EndFor
				\For{$j=r-2^k,\dotsc,2^k-1$}\label{tft:push_down_both_ii}
					\State\label{tft:rightmost_butterfly_ii}
						$\begin{pmatrix}
							\x{2^k(2q)  +j} \\
							\x{2^k(2q'+1)+j}
						\end{pmatrix}
						\set
						\begin{pmatrix}
							1 &  \alpha \\
							1 & -\alpha
						\end{pmatrix}
						\begin{pmatrix}
							\x{2^k(2q)  +j} \\
							\x{2^k(2q'+1)+j}
						\end{pmatrix}$
				\EndFor\label{tft:push_down_both_ii_end}\label{tft:push_down_both_end}
			\Else
				\For{$j=0,\dotsc,r-1$}%
					\label{tft:push_down_left_only}
					\State
						$\x{2^k(2q)+j}\set\x{2^k(2q)+j}+\alpha\x{2^k(2q'+1)+j}$
				\EndFor
				\For{$j=r,\dotsc,2^k-1$}
				\State
					$\x{2^k(2q')+j}\set\x{2^k(2q')+j}+\alpha\x{2^k(2q'+1)+j}$
				\EndFor\label{tft:push_down_left_only_end}
			\EndIf\label{tft:push_down_right_k_end}
		\EndFor\label{tft:push_down_right_end}
		\State
			\texttt{%
				// Restore $\x{\ell-2^{m-1}},\dotsc,\x{2^{m-1}-1}$ to $a_{m-1,\ell-2^{m-1}},\dotsc,a_{m-1,2^{m-1}-1}$
			}
		\For{$k=v+1,v+2,\dotsc,m-2$}\label{tft:restore}
			\State
				\label{tft:restore_k}%
				$q\set\floor{\ell/2^{k+1}}$,
				$r\set\ell-2^{k+1}q$,
				$q'\set q-2^{m-k-2}$
			\If{$r>2^k$}
				\State\label{tft:alpha_push_up_right} $\alpha\set\omega^{-[2q]_p}$
				\For{$j=r-2^k,\dotsc,2^k-1$}\label{tft:push_up_right}
					\State\label{tft:push_up_right_j}
						$\x{2^k(2q'+1)+j}
						\set
						\Div_2(\alpha(\x{2^k(2q)+j}-\x{2^k(2q'+1)+j}))$
				\EndFor\label{tft:push_up_right_end}
			\Else
				\State\label{tft:alpha_push_up_left} $\alpha\set\omega^{[2q]_p}$
				\For{$j=r,\dotsc,2^k-1$}
					\label{tft:push_up_left}
					\State
						$\x{2^k(2q')+j}
						\set
						\x{2^k(2q')+j}-\alpha\x{2^k(2q'+1)+j}$
				\EndFor\label{tft:push_up_left_end}
			\EndIf\label{tft:restore_k_end}
		\EndFor\label{tft:restore_end}
		\State
			\texttt{%
				// Compute $a_{k,0},\dotsc,a_{k,2^{k+1}\floor{\ell/2^{k+1}}-1}$ for $k=m-2,m-3,\dotsc,0$
			}
		\For{$k=m-2,m-3,\dotsc,0$}\label{tft:remainder}
			\State
				\label{tft:remainder_k}%
				\label{tft:q}%
				$q\set\floor{\ell/2^{k+1}}$
			\For{$j=0,\dotsc,2^k-1$}
			\State\label{tft:butterfly_0}
				$\begin{pmatrix}
					\x{j} \\
					\x{2^k+j}
				\end{pmatrix}
				\set
				\begin{pmatrix}
					1 &  1 \\
					1 & -1
				\end{pmatrix}
				\begin{pmatrix}
					\x{j} \\
					\x{2^k+j}
				\end{pmatrix}$
			\EndFor
			\For{$(i,\alpha)\in\{(i,\omega^{[2i]_p})\mid i\in\{1,\dotsc,q-1\}\}$}\label{tft:twiddle_set}
				\For{$j=0,\dotsc,2^k-1$}
					\State\label{tft:butterfly}
						$\begin{pmatrix}
							\x{2^k(2i  )+j} \\
							\x{2^k(2i+1)+j}
						\end{pmatrix}
						\set
						\begin{pmatrix}
							1 &  \alpha \\
							1 & -\alpha
						\end{pmatrix}
						\begin{pmatrix}
							\x{2^k(2i  )+j} \\
							\x{2^k(2i+1)+j}
						\end{pmatrix}$
				\EndFor
			\EndFor\label{tft:remainder_k_end}
		\EndFor\label{tft:remainder_end}
	\end{algorithmic}
\end{algorithm}

\subsection{Overview and correctness}

Similar to van der Hoeven's TFT algorithm, Algorithm~\ref{alg:tft} is based on the idea of computing the sequences $a_{k,0},\dotsc,a_{k,2^k\ceil{\ell/2^k}-1}$ for $k=m-1,m-2,\dotsc,0$, where $m=\ceil{\log_2\ell}$. However, the computation does not proceed by computing the full sequences for successive values of $k$, since their length is greater than $\ell$ for $k$ such that $2^k$ does not divide $\ell$. Instead, in the first phase of the algorithm, consisting of Lines~\ref{tft:m} to~\ref{tft:push_down_right_end}, the subsequences $a_{k,2^{k+1}\floor{\ell/2^{k+1}}},\dotsc,a_{k,2^k\ceil{\ell/2^k}-1}$ are computed for $k=m-1,m-2,\dotsc,v$, where $v=\ord_2\ell$ is the maximum integer for which $2^v$ divides~$\ell$. Each computed element $a_{k,i}$ is written to $\x{i}$ if $i<\ell$, and written to $\x{i-2^{m-1}}$ if $i\geq\ell$. These assignments are consistent for $k=m-1$, since Lemma~\ref{lem:level_m-1} implies that $a_{m-1,i}=a_{m-1,2^{m-1}+i}$ for ${i\in\{\ell-2^{m-1},\dotsc,2^{m-1}-1\}}$ (of course, we do not compute these elements twice), while no clashes occur for smaller values of $k$, since only elements $a_{k,i}$ with $i\geq 2^{m-1}$ are computed. This method of assigning computed elements $a_{k,i}$ to the array is used throughout the entire algorithm.

By indexing in this manner, the elements $a_{m-1,\ell-2^{m-1}},\dotsc,a_{m-1,2^{m-1}-1}$, which by Lemma~\ref{lem:level_m-1} reside in the entries $\x{\ell-2^{m-1}},\dotsc,\x{2^{m-1}-1}$ after Lines~\ref{tft:level_m-1} and~\ref{tft:level_m-1_end} have been performed, are partially overwritten by Lines~\ref{tft:push_down_right} to~\ref{tft:push_down_right_end} with elements $a_{k,i}$ such that $k\leq m-2$ and $i\geq\ell$. However, the choice of the subsequences computed in the first phase of the algorithm ensures that no element $a_{k,i}$ with $i\geq\ell$ is needed to compute the remainder of the transform, whereas no further progress can be made on computing the transform without access to $a_{m-1,\ell-2^{m-1}},\dotsc,a_{m-1,2^{m-1}-1}$. Consequently, the next phase of the algorithm, Lines~\ref{tft:restore} to~\ref{tft:restore_end}, reverts the modifications to the entries $\x{\ell-2^{m-1}},\dotsc,\x{2^{m-1}-1}$ made by Lines~\ref{tft:push_down_right} to~\ref{tft:push_down_right_end} by proceeding backwards through the loop and using~\eqref{eqn:inverse_butterfly} and~\eqref{eqn:vdH} to partially inverting the offending transformations. The final phase of the algorithm, Lines~\ref{tft:remainder} to~\ref{tft:remainder_end}, then completes the computation of the transform by computing $a_{k,0},\dotsc,a_{k,2^{k+1}\floor{\ell/2^{k+1}}-1}$ for $k=m-2,m-3,\dotsc,0$, where $2^{k+1}\floor{\ell/2^{k+1}}=2^k\ceil{\ell/2^k}$ if $k<v$.

Further details are provided by the proof of the following lemma.

\begin{lemma}\label{lem:correctness} Algorithm~\ref{alg:tft} correctly evaluates $\TFT_{\omega,\ell}$.
\end{lemma}
\begin{proof} Suppose that Algorithm~\ref{alg:tft} is called with $\vec{a}=(a_0,\dotsc,a_{\ell-1})\in\A^\ell$. Define $a_{k,0},\dotsc,a_{k,n-1}$ for $k\in\{0,\dotsc,p\}$ as per Lemma~\ref{lem:radix-2_FFT}, with $a_i$ taken to be zero for $i\in\{\ell,\dotsc,n-1\}$. Then the lemma implies that $\TFT_{\omega,\ell}(\vec{a})$ is computed correctly by Algorithm~\ref{alg:tft} if $\x{i}=a_{i,0}$ for $i\in\{0,\dotsc,\ell-1\}$ upon its termination. 

Define $m=\ceil{\log_2\ell}$ and $v=\ord_2\ell$, as in Line~\ref{tft:m} of Algorithm~\ref{alg:tft}. Then Lemma~\ref{lem:level_m-1} implies $\x{i}=a_{m-1,i}$ for $i\in\{0,\dotsc,\ell-1\}$ after Lines~\ref{tft:level_m-1} and~\ref{tft:level_m-1_end} of the algorithm have been performed. If $\ell=2^m$, the algorithm then proceeds directly to the loop of Lines~\ref{tft:remainder} to~\ref{tft:remainder_end}. In this case, it follows immediately from the definition of the $a_{k,i}$ that each iteration of the loop sets $\x{i}=a_{k,i}$ for $i\in\{0,\dotsc,\ell-1\}$, where $k$ is the corresponding value of the loop variable. Thus, the algorithm terminates with $\x{i}=a_{i,0}$ for $i\in\{0,\dotsc,\ell-1\}$ in this case. Therefore, we may assume that $\ell<2^m$. Then $v\leq m-2$, since it is never equal to $m-1$.

For $k\in\{0,\dotsc,m-1\}$, define $q_k=\floor{\ell/2^{k+1}}$, $r_k=\ell-2^{k+1}q_k$ and $q_k'=q_k-2^{m-2-k}$, so as to reflect the values $q$, $r$ and $q'$ computed in Lines~\ref{tft:q_and_r}, \ref{tft:restore_k} and~\ref{tft:q} of the algorithm. Furthermore, define index sets
\begin{equation*}
	R_k
	=
	\begin{cases}
		\{
			i\in\Z
			\mid
			2^{m-1}\leq i<\ell
		\}
		& \text{if $k=m-1$}, \\
		\{
			i\in\Z
			\mid
			2^k(2q_k)\leq i<\ell
		\}
		& \text{if $k< m-1$},
	\end{cases}
\end{equation*}
and
\begin{equation*}
	R_k'
	=
	\begin{cases}
		\{
			i\in\Z
			\mid
			\ell-2^{m-1}\leq i<2^k(2q_k'+2)
		\}
		&\text{if $r_k>2^k$},\\
		\{
			i\in\Z
			\mid
			\ell-2^{m-1}\leq i<2^k(2q_k'+1)
		\}
		&\text{if $r_k\leq 2^k$}.
	\end{cases}
\end{equation*}
Then in Lines~\ref{tft:push_down_right} to~\ref{tft:push_down_right_end} of the algorithm (and also Lines~\ref{tft:restore} to~\ref{tft:restore_end}), indices of the form $2^k(2q)+j$ or $2^k(2q+1)+j$ always belong to $R_{m-1}=\{2^{m-1},\dotsc,\ell-1\}$, while indices of the form $2^k(2q')+j$ or $2^k(2q'+1)+j$ always belong to $R_{m-1}'=\{\ell-2^{m-1},\dotsc,2^{m-1}-1\}$. Moreover, for each value of the loop variable $k$, the set $R_k\cup R_k'$ consists of the indices $i$ for which $\x{i}$ is overwritten during the corresponding iteration of the loop. For $k\in\{v,\dotsc,m-2\}$,
\begin{equation}\label{eqn:q_k-recursion}
	q_k
	=
	\begin{cases}
		2q_{k+1}+1 & \text{if $r_{k+1}>2^{k+1}$},     \\
		2q_{k+1}   & \text{if $r_{k+1}\leq 2^{k+1}$}.
	\end{cases}
\end{equation}
Consequently, $R_v\subseteq R_{v+1}\subseteq\dotsb\subseteq R_{m-1}$ and $\emptyset=R_v'\subseteq R_{v+1}'\subseteq\dotsb\subseteq R_{m-1}'$.

We now use induction to determine the effect of the loop in Lines~\ref{tft:push_down_right} to~\ref{tft:push_down_right_end}. As $\x{i}=a_{m-1,i}$ for $i\in\{0,\dotsc,\ell-1\}$ after Lines~\ref{tft:level_m-1} and~\ref{tft:level_m-1_end} have been performed, Lemma~\ref{lem:level_m-1} implies that $\x{i}=a_{m-1,i}$ for $i\in R_{m-1}$, and $\x{i}=a_{m-1,i}=a_{m-1,2^{m-1}+i}$ for $i\in R_{m-1}'$, at the beginning of the first iteration of the loop. Suppose more generally that $\x{i}=a_{k+1,i}$ for $i\in R_{k+1}$, and $\x{i}=a_{k+1,2^{m-1}+i}$ for $i\in R_{k+1}'$, at the beginning of some iteration of the loop, where $k\in\{v,\dotsc,m-2\}$ is the corresponding value of the loop variable. If $r_k>2^k$, then
\begin{equation*}
	a_{k+1,2^k(2q_k)+j}
	=
	\x{2^k(2q_k)+j}
	\quad\text{and}\quad
	a_{k+1,2^k(2q_k+1)+j}
	=
	\begin{cases}
		\x{2^k(2q_k+1)+j}  &\text{if $j<r_k-2^k$},    \\
		\x{2^k(2q_k'+1)+j} &\text{if $j\geq r_k-2^k$},
	\end{cases}
\end{equation*}
for $j\in\{0,\dotsc,2^k-1\}$, since $R_k\subseteq R_{k+1}$ and $R_k'\subseteq R_{k+1}'$.
Thus, \eqref{eqn:butterfly} implies that Lines~\ref{tft:push_down_both} to~\ref{tft:push_down_both_end} compute $a_{k,2^k(2q_k)+j}$ and $a_{k,2^k(2q_k+1)+j}$ for $j\in\{0,\dotsc,2^k-1\}$, with $a_{k,i}$ written to $\x{i}$ for $i\in R_k$, and $a_{k,2^{m-1}+i}$ written to $\x{i}$ for $i\in R_k'$. If $r_k\leq 2^k$, then
\begin{equation*}
	a_{k+1,2^k(2q_k)+j}
	=
	\begin{cases}
		\x{2^k(2q_k)+j}  &\text{if $j<r_k$},     \\
		\x{2^k(2q_k')+j} &\text{if $j\geq r_k$},
	\end{cases}
	\quad\text{and}\quad
	a_{k+1,2^k(2q_k+1)+j}
	=
	\x{2^k(2q_k'+1)+j}
\end{equation*}
for $j\in\{0,\dotsc,2^k-1\}$, since $R_k\subseteq R_{k+1}$, $R_k'\subseteq R_{k+1}'$ and~\eqref{eqn:q_k-recursion} implies that
\begin{equation}\label{eqn:R_prime_setminus}
	\left\{
		i\in\Z
		\mid
		2^k(2q_k'+1)\leq i< 2^k(2q_k'+1)+2^k
	\right\}
	=
	R_{k+1}'\setminus R_k'.
\end{equation}
Thus, \eqref{eqn:butterfly} implies that Lines~\ref{tft:push_down_left_only} to~\ref{tft:push_down_left_only_end} compute $a_{k,2^k(2q_k)+j}$ for $j\in\{0,\dotsc,2^k-1\}$, with $a_{k,i}$ written to~$\x{i}$ for $i\in R_k$, and $a_{k,2^{m-1}+i}$ written to $\x{i}$ for $i\in R_k'$. It follows that each iteration of the loop in Lines~\ref{tft:push_down_right} to~\ref{tft:push_down_right_end} sets $\x{i}=a_{k,i}$ for $i\in R_k$, and $\x{i}=a_{k+1,2^{m-1}+i}$ for $i\in R_k'$, where $k$ is the corresponding value of the loop variable. Therefore, after Lines~\ref{tft:push_down_right} to~\ref{tft:push_down_right_end} have been performed, the state of the array $(\x{0},\dotsc,\x{\ell-1})$ is summarised as follows:
\begin{equation}\label{eqn:after_push_down_right}
	\x{i}=
	\begin{cases}
		a_{m-1,i}         & \text{if $i\notin R_{m-1}\cup R_{m-1}'$}, \\
		a_{k+1,i}         & \text{if $i\in R_{k+1}\setminus R_k$ for some $k\geq v$},    \\
		a_{k+1,i+2^{m-1}} & \text{if $i\in R_{k+1}'\setminus R_k'$ for some $k\geq v$},  \\
		a_{v,i}           & \text{if $i\in R_v$}.
	\end{cases}
\end{equation}

We now use induction to determine the effect of the loop in Lines~\ref{tft:restore} to~\ref{tft:restore_end}. Therefore, suppose that
$\x{i}=a_{k,2^{m-1}+i}$ for $i\in R_{k-1}'$ at the beginning of some iteration of the loop, where $k\in\{v+1,\dotsc,m-2\}$ is the corresponding value of the loop variable. Then $\x{i}=a_{k,2^{m-1}+i}$ for $i\in R_{k-1}'\cup(R_k'\setminus R_{k-1}')=R_k'$, $\x{i}=a_{k,i}$ for $i\in R_k\setminus R_{k-1}$, and $\x{i}=a_{k+1,2^{m-1}+i}$ for $i\in R_{k+1}'\setminus R_k'$, since \eqref{eqn:after_push_down_right} held after Lines~\ref{tft:push_down_right} to~\ref{tft:push_down_right_end} were performed, while any previous iterations of the loop only modified~$\x{i}$ for $i\in R_{v+1}'\cup\dotsb\cup R_{k-1}'=R_{k-1}'$. Thus, if $r_k>2^k$, then $\x{2^k(2q_k'+1)+j}=a_{k,2^k(2q_k+1)+j}$ and $\x{2^k(2q_k)+j}=a_{k,2^k(2q_k)+j}$ for $j\in\{r_k-2^k,\dotsc,2^k-1\}$, since~\eqref{eqn:q_k-recursion} implies that
\begin{equation}\label{eqn:R_setminus}
	R_k\setminus R_{k-1}
	=
	\left\{
		i\in\Z
		\mid
		2^k(2q_k)\leq i<2^k(2q_k)+2^k
	\right\}.
\end{equation}
Consequently, \eqref{eqn:inverse_butterfly} implies in this case that Lines~\ref{tft:push_up_right} and~\ref{tft:push_up_right_end} set $\x{i}=a_{k+1,2^{m-1}+i}$ for $i\in R_k'$. If $r_k\leq 2^k$, then $\x{2^k(2q_k')+j}=a_{k,2^k(2q_k)+j}$ and $\x{2^k(2q_k'+1)+j}=a_{k+1,2^k(2q_k+1)+j}$ for $j\in\{r_k,\dotsc,2^k-1\}$, since~\eqref{eqn:R_prime_setminus} holds. Thus, \eqref{eqn:vdH} implies that Lines~\ref{tft:push_up_left} and~\ref{tft:push_up_left_end} set $\x{i}=a_{k+1,2^{m-1}+i}$ for $i\in R_k'$. Therefore, we find that if $\x{i}=a_{k,2^{m-1}+i}$ for $i\in R_{k-1}'$ at the beginning of an iteration of the loop in Lines~\ref{tft:restore} to~\ref{tft:restore_end}, then $\x{i}=a_{k+1,2^{m-1}+i}$ for $i\in R_k'$ at the end of the iteration. As $R_v'$ is the empty set, it follows that the overall effect of loop is to set $\x{i}=a_{m-1,2^m+i}$ for $i\in R_{m-2}'$. Hence, $\x{i}=a_{m-1,2^m+i}=a_{m-1,i}$ for $i\in(R_{m-1}'\setminus R_{m-2}')\cup R_{m-2}'=R_{m-1}'$ after the loop has been performed.

Redefine $R_{m-1}$ to be the set $\{0,\dotsc,\ell-1\}$. Then, as $R_0=\dotsb=R_{v-1}=\emptyset$ if $v>0$, the state of the vector $(\x{0},\dotsc,\x{\ell-1})$ after Lines~\ref{tft:restore} to~\ref{tft:restore_end} have been performed is summarised as follows:
\begin{equation}\label{eqn:after_swap}
	\x{i}
	=
	\begin{cases}
		a_{k+1,i} & \text{if $i\in R_{k+1}\setminus R_k$}, \\
		a_{0,i}   & \text{if $i\in R_0$}.
	\end{cases}
\end{equation}
For $k\in\{0,\dotsc,m-1\}$, define $L_k=\{i\in\Z\mid 0\leq i<2^k(2q_k)\}$. Then for each value of the loop variable $k$ in Line~\ref{tft:remainder}, the set $L_k$ consists of the indices~$i$ for which~$\x{i}$ is modified by Lines~\ref{tft:remainder_k} to~\ref{tft:remainder_k_end} during the corresponding iteration of the loop. Moreover, $L_k$ is the complement of $R_k$ in $\{0,\dotsc,\ell-1\}$ for $k\in\{0,\dotsc,m-1\}$, and $L_0\supseteq L_1\supseteq\dotsb\supseteq L_{m-1}=\emptyset$.

Suppose that $\x{i}=a_{k+1,i}$ for $i\in L_{k+1}$ at the beginning of an iteration of the loop in Lines~\ref{tft:remainder} to~\ref{tft:remainder_end}, where $k$ is the corresponding value of the loop variable. Then $\x{i}=a_{k+1,i}$ for $i\in L_{k+1}\cup(R_{k+1}\setminus R_k)=L_k$, since~\eqref{eqn:after_swap} held after Lines~\ref{tft:restore} to~\ref{tft:restore_end} were performed, while any previous iterations of the loop only modified~$\x{i}$ for $i\in L_{k+1}\cup\dotsb\cup L_{m-1}=L_{k+1}$. Thus, \eqref{eqn:butterfly} implies that $\x{i}=a_{k,i}$ for $i\in L_k$ at the end of the iteration. As $L_{m-1}=\emptyset$, it holds trivially that $\x{i}=a_{m-1,i}$ for $i\in L_{m-1}$ for at the beginning of the first iteration of the loop in Lines~\ref{tft:remainder} to~\ref{tft:remainder_end}. Therefore, $\x{i}=a_{0,i}$ for $i\in L_0$ once all iterations of the loop have been performed. However, it also holds that $\x{i}=a_{0,i}$ for $i\in R_0=\{0,\dotsc,\ell-1\}\setminus L_0$, since~\eqref{eqn:after_swap} held after Lines~\ref{tft:restore} to~\ref{tft:restore_end} were performed, while Lines~\ref{tft:remainder} to~\ref{tft:remainder_end} only modified $\x{i}$ for $i\in L_0\cup\dotsb\cup L_{m-2}=L_0$. It follows that the algorithm terminates with $\x{i}=a_{i,0}$ for $i\in\{0,\dotsc,\ell-1\}$, as required.
\end{proof}

\subsection{Avoiding division by two}\label{sec:avoiding_inv_2}

It is possible to modify Algorithm~\ref{alg:tft} so as to avoid the use of the map $\Div_2$, and without incurring any additional costs.
If Lines~\ref{tft:push_down_both_ii} and~\ref{tft:push_down_both_ii_end} are performed for some $k\in\{v,\dotsc,m-2\}$, then the values stored in the entries $\x{2^k(2q)+j}$ for $j\in\{r-2^k,\dotsc,2^k-1\}$ by the loop may be modified, so long as they are then set to the correct values when Lines~\ref{tft:push_up_right} and~\ref{tft:push_up_right_end} are performed for the same value of~$k$. Using the notation of the proof of Lemma~\ref{lem:correctness}, we see that such a modification is possible since only entries $\x{i}$ with indices belonging to $R_{k-1}\cup R_{k-1}'$ are accessed or modified during the intervening period, while~\eqref{eqn:R_setminus} implies that the entries $\x{2^k(2q_k)+j}$ for $j\in\{r_k-2^k,\dotsc,2^k-1\}$ have indices that belong to $R_k\setminus R_{k-1}=R_k\setminus(R_{k-1}\cup R_{k-1}')$. Consequently, the algorithm may be modified so that Line~\ref{tft:rightmost_butterfly_ii} replaces $(\x{2^k(2q)+j},\x{2^k(2q'+1)+j})^t$ by $A(\x{2^k(2q)+j},\x{2^k(2q'+1)+j})^t$, and Line~\ref{tft:push_up_right_end} replaces $(\x{2^k(2q)+j},\x{2^k(2q'+1)+j})^t$ by $B(\x{2^k(2q)+j},\x{2^k(2q'+1)+j})^t$, for $2\times 2$ matrices $A$ and $B$ over $\A$ such that
\begin{equation*}
	\begin{pmatrix}
		0 & 1
	\end{pmatrix}
	A
	=
	\begin{pmatrix}
		1 & -\omega^{[2q]_p}
	\end{pmatrix}
	\quad\text{and}\quad
	BA
	=
	\begin{pmatrix}
		1 & \omega^{[2q]_p} \\
		0 & 1 
	\end{pmatrix}.
\end{equation*}
Indeed, the requirement on the bottom row of $A$ ensures that Line~\ref{tft:rightmost_butterfly_ii} assigns the correct value to $\x{2^k(2q'+1)+j}$, while the requirement on the product $BA$ ensures that Line~\ref{tft:push_up_right_end} reverts the effect of Line~\ref{tft:rightmost_butterfly_ii} on $\x{2^k(2q'+1)+j}$, and corrects the value stored in $\x{2^k(2q)+j}$.

If $2$ is a unit in $\A$, then Algorithm~\ref{alg:tft} in its current form is equivalent to taking
\begin{equation*}
	A
	=
	\begin{pmatrix}
		1 &  \omega^{[2q]_p} \\
		1 & -\omega^{[2q]_p} 
	\end{pmatrix}
	\quad\text{and}\quad
	B
	=
	\begin{pmatrix}
		1                      & 0                       \\
		2^{-1}\omega^{-[2q]_p} & -2^{-1}\omega^{-[2q]_p}
	\end{pmatrix}.
\end{equation*}
To remove the use of the map $\Div_2$, along with the assumption that $2$ is not a zero-divisor in $\A$, we consider the choice
\begin{equation*}
	A
	=
	\begin{pmatrix}
		0 & 1               \\
		1 & -\omega^{[2q]_p} 
	\end{pmatrix}
	\quad\text{and}\quad
	B
	=
	\begin{pmatrix}
		2\omega^{[2q]_p} & 1 \\
		1                & 0
	\end{pmatrix},
\end{equation*}
which implies that Algorithm~\ref{alg:tft} may be modified as follows: replace Line~\ref{tft:rightmost_butterfly_ii} by
\begin{equation}\label{eqn:modified_push_down_both_end}
	\begin{pmatrix}
		\x{2^k(2q)  +j} \\
		\x{2^k(2q'+1)+j}
	\end{pmatrix}
	\set
	\begin{pmatrix}
		0 &  1 \\
		1 & -\alpha
	\end{pmatrix}
	\begin{pmatrix}
		\x{2^k(2q)  +j} \\
		\x{2^k(2q'+1)+j}
	\end{pmatrix},
\end{equation}
replace Line~\ref{tft:alpha_push_up_right} by $\alpha\set\omega^{[2q]_p}$, and replace Line~\ref{tft:push_up_right_end} by
\begin{equation}\label{eqn:modified_push_up_right_end}
	\begin{pmatrix}
		\x{2^k(2q)  +j} \\
		\x{2^k(2q'+1)+j}
	\end{pmatrix}
	\set
	\begin{pmatrix}
		2\alpha &  1 \\
		1       &  0
	\end{pmatrix}
	\begin{pmatrix}
		\x{2^k(2q)  +j} \\
		\x{2^k(2q'+1)+j}
	\end{pmatrix}.
\end{equation}
Performing~\eqref{eqn:modified_push_down_both_end} rather than Line~\ref{tft:rightmost_butterfly_ii} saves an addition, while performing~\eqref{eqn:modified_push_up_right_end} instead of Line~\ref{tft:push_up_right_end} replaces the cost of evaluating the map $\Div_2$ with that of multiplying by $2$.

\subsection{Computing the twiddle factors}\label{sec:twiddle_factors}

We now describe how the twiddle factors are computed in Algorithm~\ref{alg:tft}. Therefore, suppose that the algorithm has been called to compute a length $\ell$ transform. Let $m=\ceil{\log_2\ell}$ and $v=\ord_2\ell$, as defined in Line~\ref{tft:m} of the algorithm. Furthermore, let $q_k=\floor{\ell/2^{k+1}}$ for $k\in\{0,\dotsc,m-1\}$, so as to reflect the values of $q$ computed in Lines~\ref{tft:q_and_r}, \ref{tft:restore_k} and~\ref{tft:q}.

For the computation of the twiddle factors in Lines~\ref{tft:push_down_right_k}, \ref{tft:alpha_push_up_right} and~\ref{tft:alpha_push_up_left} we may assume that $v<m$, since otherwise the lines are not performed by the algorithm. Then, for $k\in\{v,\dotsc,m-1\}$, we have
\begin{equation*}\label{eqn:alg_tft_twiddle_factors}
	\omega^{[2q_k]_p}
	=
	\left(\omega^{n/2^m}\right)^{2^k[q_k]_{m-k-1}}
	\quad\text{and}\quad
	\omega^{-[2q_k]_p}
	=
	\left(\omega^{n/2^m}\right)^{2^k(2^{m-k}-[q_k]_{m-k-1})},
\end{equation*}
since $q_k<2^{m-k-1}$. As $\omega^{n/2^{\ceil{\log_2\ell}}}=\omega^{n/2^m}$ is provided as an input to Algorithm~\ref{alg:tft}, it follows that the twiddle factors of Lines~\ref{tft:push_down_right_k} and~\ref{tft:alpha_push_up_left}, and the inverse twiddle factors of Line~\ref{tft:alpha_push_up_right}, can be computed by a square and multiply algorithm (see~\cite{knuth1997}) with $\bigO(1)$ space and an individual cost of $\bigO(m)$ multiplications by powers of~$\omega$. We propose to use this approach for Lines~\ref{tft:push_down_right_k}, \ref{tft:alpha_push_up_right} and~\ref{tft:alpha_push_up_left} of the algorithm, at a total cost of $\bigO(\log^2\ell)$ multiplications by powers of $\omega$. However, we also note that Lines~\ref{tft:alpha_push_up_right} and~\ref{tft:alpha_push_up_left} can be implemented with a total cost of $\bigO(\log\ell)$ multiplications by powers of $\omega$, since~\eqref{eqn:q_k-recursion} implies that
\begin{equation*}
	\omega^{[2q_{k+1}]_p}
	=
	\begin{cases}
		-\left(\omega^{[2q_k]_p}\right)^2
		& \text{if $r_{k+1}>2^{k+1}$},\\
		\left(\omega^{[2q_k]_p}\right)^2
		& \text{if $r_{k+1}\leq 2^{k+1}$}, 
	\end{cases}
\end{equation*}
for $k\in\{v,\dotsc,m-2\}$.

Line~\ref{tft:twiddle_set} asks that we compute the pairs $(i,\omega^{[2i]_p})$ for $i\in\{1,\dotsc,q-1\}$, where $q\in\{1,\dotsc,2^{m-1}\}$. We implement the line by generalising the approach of Section~\ref{sec:FFT} in a manner that is reminiscent of the decomposition of Lemma~\ref{lem:tft_decomposition}. Therefore, suppose that $q\in\{1,\dotsc,2^{m-1}\}$. Then $q$ may be written in the form $\sum^t_{r=1}2^{i_r}$ for $t\geq 1$ nonnegative integers $i_1>i_2>\dotsb>i_t$. The sets $\{2^{i_1}+\dotsb+2^{i_{s-1}}+[j]_{i_s}\mid j\in\{0,\dotsc,2^{i_s}-1\}\}$ for $s\in\{1,\dotsc,t\}$ then form a partition of $\{0,\dotsc,q-1\}$. Furthermore,
\begin{equation*}
	[2(2^{i_1}+\dotsb+2^{i_{s-1}}+[j]_{i_s})]_p
	=
	2^{p-1-i_s}j
	+
	\sum^{s-1}_{r=1}
	2^{p-2-i_r}
\end{equation*}
for $s\in\{1,\dotsc,t\}$ and $j\in\{0,\dotsc,2^{i_s}-1\}$. Thus, if $q<2^{m-1}$, then $\omega^{[2(2^{i_1}+\dotsb+2^{i_{s-1}}+[j]_{i_s})]_p}=\tau_s(\lambda^2_s)^{j}$ for $s\in\{1,\dotsc,t\}$ and $j\in\{0,\dotsc,2^{i_s}-1\}$, where $\lambda_1,\dotsc,\lambda_t,\tau_1,\dotsc,\tau_t\in\A$ are defined recursively as follows:
\begin{equation*}
	\lambda_r
	=
	\begin{cases}
		\left(\omega^{n/2^m}\right)^{2^{m-2-i_1}}        & \text{if $r=1$},\\
		\lambda^{2^{i_{r-1}-i_r}}_{r-1} & \text{if $r>1$},
	\end{cases}
	\quad\text{and}\quad
	\tau_r
	=
	\begin{cases}
	1                       & \text{if $r=1$},\\
	\tau_{r-1}\lambda_{r-1} & \text{if $r>1$}.
	\end{cases}
\end{equation*}
If $q=2^{m-1}$, then $t=1$ and $\omega^{[2[j]_{i_1}]_p}=(\omega^{n/2^m})^j$ for $j\in\{0,\dotsc,2^{i_1}-1\}$. Therefore, computing the pairs $(i,\omega^{[2i]_p})$ for $i\in\{1,\dotsc,q-1\}$ reduces to computing the terms of $t$ geometric progressions whose first terms and common ratios are related by simple recurrence relations. We derive Algorithm~\ref{alg:pairs} from these observations, which is presented in the style of a generator that may be called by Algorithm~\ref{alg:tft}. Assuming that repeated squaring is used for exponentiation, the algorithm performs $q+\bigO(m)$ multiplications by powers of~$\omega$, and uses $\bigO(1)$ space.

\begin{algorithm}[ht]
	\caption{$\textproc{GeneratePairs}(m,\omega^{n/2^m},q)$}
	\label{alg:pairs}
	\begin{algorithmic}[1]
		\Require $m\in\{1,\dotsc,p\}$, the precomputed element $\omega^{n/2^m}$, and $q\in\{1,\dotsc,2^{m-1}\}$.
		\renewcommand{\algorithmicensure}{\textbf{Output:}}
		\Ensure the pairs $(i,\omega^{[2i]_p})$ for $i\in\{1,\dotsc,q-1\}$, yielded one at a time.
		\State $i\set\floor{\log_2 q}$, $\tau\set 1$, $e\set\min(m-1-i,1)$, $\lambda\set(\omega^{n/2^m})^{2^{m-1-i-e}}$, $\mu\set\lambda^{2^e}$, $\theta\set\tau$
		\For{$j=1,\dotsc,2^i-1$}
			\State $\theta\set\theta\mu$
			\State \textbf{yield} $([j]_i,\theta)$
		\EndFor
		\State $o\set 2^i$
		\While{$q>o$}
			\State $i'\set i$, $i\set\floor{\log_2(q-o)}$, $\tau\set\tau\lambda$, $\lambda\set\lambda^{2^{i'-i}}$, $\mu\set\lambda^2$, $\theta\set\tau$
			\State \textbf{yield} $(o,\theta)$
			\For{$j=1,\dotsc,2^i-1$}
				\State $\theta\set\theta\mu$
				\State \textbf{yield} $(o+[j]_i,\theta)$
			\EndFor
			\State $o\set o+2^i$
		\EndWhile
	\end{algorithmic}
\end{algorithm}

\subsection{Complexity}

Recall that Lines~\ref{tft:m} to~\ref{tft:push_down_right_end} of Algorithm~\ref{alg:tft} compute $a_{k,2^{k+1}\floor{\ell/2^{k+1}}},\dotsc,a_{k,2^k\ceil{\ell/2^k}-1}$ for $k=m-1,m-2,\dotsc,v$, while Lines~\ref{tft:remainder} to~\ref{tft:remainder_end} compute $a_{k,0},\dotsc,a_{k,2^{k+1}\floor{\ell/2^{k+1}}-1}$ for $k=m-2,m-3,\dotsc,0$. Thus, the combined effect of these lines is to perform van der Hoeven's TFT algorithm. It follows that the difference in complexities between Algorithm~\ref{alg:tft} and van der Hoeven's algorithm is equal to cost of Lines~\ref{tft:restore} to~\ref{tft:restore_end}, and the computation of the twiddle factors, which are precomputed in van der Hoeven's algorithm. If the modifications of Section~\ref{sec:avoiding_inv_2} are applied, then the difference is bounded by the costs of Lines~\ref{tft:restore} to~\ref{tft:restore_end}, and the computation of the twiddle factors. We now show that these costs amount to only $\bigO(\ell)$ operations in $\A$ for the modified algorithm, so that the difference in complexities between the algorithm and van der Hoeven's algorithm, and thus also Arnold's algorithm, is confined to lower order terms.

\begin{lemma}\label{lem:tft_complexity} Suppose that the modifications of Section~\ref{sec:avoiding_inv_2} are applied to Algorithm~\ref{alg:tft} with the multiplication by $2$ in~\eqref{eqn:modified_push_up_right_end} implemented as an addition. Then the algorithm performs at most $\sum^t_{r=1}(\ell_r/2)\log_2\ell_r+2\ell+\bigO(\log^2\ell)$ multiplications by powers of $\omega$, and at most $\ell\floor{\log_2\ell}+2\ell$ additions or subtractions, where $\ell_1>\ell_2>\dotsb>\ell_t$ are integer powers of two such that $\ell=\sum^t_{r=1}\ell_r$.
\end{lemma}
\begin{proof} Suppose that the modified algorithm has been called to evaluate a length $\ell$ transform. Let $m=\ceil{\log_2\ell}$ and $v=\ord_2\ell$, as defined in Line~\ref{tft:m} of the algorithm. Furthermore, define $q_k=\floor{\ell/2^{k+1}}$ and $r_k=\ell-2^{k+1}q_k$ for $k\in\{0,\dotsc,m-1\}$, so as to reflect the values of $q$ and $r$ computed in Lines~\ref{tft:q_and_r}, \ref{tft:restore_k} and~\ref{tft:q}. Write $\ell=\sum^m_{k=0}2^k\ell_k$ with $\ell_0,\dotsc,\ell_m\in\{0,1\}$, which we note is different to how~$\ell$ is written in the statement of the lemma. All multiplications in the algorithm involve a power of $\omega$. Consequently, we simply refer to them as multiplications hereafter. Similarly, we do not distinguish between additions and subtractions, since they are counted together.
	
Lines~\ref{tft:level_m-1} and~\ref{tft:level_m-1_end} perform no multiplications and $2(\ell-2^{m-1})$ additions. The definition of $v$ implies that $\ell_v=1$ and $r_v=2^v$. Thus, Lines~\ref{tft:push_down_right_k} to~\ref{tft:push_down_right_k_end} perform at most $2^v$ additions for $k=v$. It also follows that for $k\in\{v+1,\dotsc,m-2\}$, $r_k>2^k$ if and only if $\ell_k=1$. Consequently, for $k\in\{v+1,\dotsc,m-2\}$ such that $r_k>2^k$, Lines~\ref{tft:push_down_right_k} to~\ref{tft:push_down_right_k_end} and Lines~\ref{tft:restore_k} to~\ref{tft:restore_k_end} perform a combined total of $r_k+2(2^{k+1}-r_k)=2^{k+1}+(2^{k+1}-r_k)\leq 2^{k+1}+2^k\ell_k$ additions. Similarly, for $k\in\{v+1,\dotsc,m-2\}$ such that $r_k\leq 2^k$, Lines~\ref{tft:push_down_right_k} to~\ref{tft:push_down_right_k_end} and Lines~\ref{tft:restore_k} to~\ref{tft:restore_k_end} perform a combined total $2^k+(2^k-r_k)\leq 2^{k+1}+2^k\ell_k$ additions. Therefore, Lines~\ref{tft:push_down_right} to~\ref{tft:restore_end} perform at most
\begin{equation*}
	2^v
	+
	\sum^{m-2}_{k=v+1}
	\left(2^{k+1}+2^k\ell_k\right)
	\leq
	\ell+2^{m-1}
\end{equation*}
additions. Excluding the computation of the twiddle factors, the number of multiplications by performed by Lines~\ref{tft:push_down_right} to~\ref{tft:restore_end} is easily seen to be bounded by the number of additions they perform. Moreover, we know from Section~\ref{sec:twiddle_factors} that the total cost of computing the twiddle factors in Lines~\ref{tft:push_down_right} to~\ref{tft:restore_end} is $\bigO(\log^2\ell)$ multiplications. It follows that Lines~\ref{tft:push_down_right} to~\ref{tft:restore_end} perform at most $\ell+2^{m-1}+\bigO(\log^2\ell)$ multiplications.

For each $k\in\{0,\dotsc,m-2\}$, Lines~\ref{tft:remainder_k} to~\ref{tft:remainder_k_end} perform $2^k(q_k-1)+q_k+\bigO(\log\ell)$ multiplications, and $2^{k+1}q_k$ additions, since we know from Section~\ref{sec:twiddle_factors} that  Line~\ref{tft:twiddle_set} performs $q_k+\bigO(\log\ell)$ multiplications. We have
\begin{equation*}
	\sum^{m-2}_{k=0}2^{k+1}q_k
	=
	\sum^{m-2}_{k=0}
	\sum^m_{i=k+1}2^i\ell_i
	=
	\sum^m_{i=0}2^ii\ell_i
	-2^m\ell_m
	\leq
	\ell\floor{\log_2\ell}
	-
	(\ell-2^{m-1})
\end{equation*}
and
\begin{equation*}
	\sum^{m-2}_{k=0}
	\left(q_k-2^k\right)
	\leq
	\sum^{m-2}_{k=0}
	\left(
		\frac{\ell}{2^{k+1}}
		-
		2^k
	\right)
	\leq
	\ell-2^{m-1}+1.
\end{equation*}
Thus, Lines~\ref{tft:remainder} to~\ref{tft:remainder_end} perform at most $\sum^m_{i=0}2^{i-1}i\ell_i+\ell-2^{m-1}+\bigO(\log^2\ell)$ multiplications, and at most $\ell\floor{\log_2\ell}-(\ell-2^{m-1})$ additions. Summing each set of bounds and discarding the terms with $\ell_i=0$ then completes the proof.
\end{proof}

Algorithm~\ref{alg:tft} and Arnold's TFT algorithm both reduce to the algorithm of Section~\ref{sec:FFT} (without the permutation) when $\ell$ is a power of two. If $\ell$ is not a power of two, then Arnold's algorithm performs approximately $\sum^t_{r=1}(\ell_r/2)\log_2\ell_r$ multiplications, where the error made is $\bigO(\log^2\ell)$, in order to evaluate the TFTs of the decomposition in Lemma~\ref{lem:tft_decomposition}. To compute the polynomials $f_r(x/\omega^{e_r})$ that are evaluated by the transforms, Arnold's algorithm first replaces the coefficients of $f$ by those of $f(\omega^{[\ell]_p}x)$ ``term-by-term'', by computing successive powers of $\omega^{[\ell]_p}$, and multiplying each power by the corresponding coefficient of $f$. By observing that $\omega^{[\ell]_p\ell_1}=-1$, one only has to compute the first $\ell_1$ powers of $\omega^{[\ell]_p}$. However, the observation also implies that these $\ell_1>\ell/2$ powers are distinct. Thus, without further improvements, Arnold's algorithm performs at least $1.5\ell$ multiplications when computing the coefficients of $f(\omega^{[\ell]_p}x)$. The coefficients of the residues $f_r$ are later replaced with those of $f_r(x/\omega^{e_r})$ by a similar method, requiring at least $\ell-\ceil{\log_2\ell}$ multiplications. Thus, Lemma~\ref{lem:tft_complexity} implies that Algorithm~\ref{alg:tft} performs $\Omega(\ell)$ fewer multiplications than Arnold's algorithm when $\ell$ is not a power of two. 

\section{An in-place inverse truncated Fourier transform}\label{sec:itft}

Assume for simplicity that $2$ is a unit in $\A$. Then Algorithm~\ref{alg:tft} is comprised of a series of invertible linear transformations. Consequently, the truncated Fourier transform can be inverted in this case by proceeding backwards through the algorithm and inverting each transformation. However, this approach does not immediately yield an algorithm with complexity similar to that of Algorithm~\ref{alg:tft}, since the transformations of  Lines~\ref{tft:level_m-1_butterfly}, \ref{tft:rightmost_butterfly_i}, \ref{tft:rightmost_butterfly_ii}, \ref{tft:butterfly_0} and~\ref{tft:butterfly} are less expensive to evaluate than their inverses:
\begin{equation*}
	\begin{pmatrix}
		1 &  \alpha \\
		1 & -\alpha
	\end{pmatrix}^{-1}
	=
	2^{-1}
	\begin{pmatrix}
		1           & 1            \\
		\alpha^{-1} & -\alpha^{-1}
	\end{pmatrix}
	\quad\text{for $\alpha\in\A^\times$}.
\end{equation*}

We can address this problem when $\ell$ is a power of two, in which case it is only necessary to invert the transformations of Lines~\ref{tft:level_m-1_butterfly}, \ref{tft:butterfly_0} and \ref{tft:butterfly}, by simply ignoring multiplications by $2^{-1}$ during the algorithm, and accounting for them afterwards by multiplying each entry of $(\x{0},\dotsc,\x{\ell-1})$ by $(2^{-1})^{\log_2\ell}$. The algorithm then performs only $\ell+\bigO(\log\log\ell)$ more multiplications than Algorithm~\ref{alg:tft}. To generalise this approach to arbitrary $\ell$, we weight the sequences defined in Lemma~\ref{lem:radix-2_FFT} by setting $\tilde{a}_{k,i}=2^ka_{k,i}$ for $k\in\{0,\dotsc,p\}$ and $i\in\{0,\dotsc,n-1\}$. Then~\eqref{eqn:butterfly} implies that
\begin{equation*}
	\begin{pmatrix}
		2\tilde{a}_{k,2^k(2i  )+j} \\
		2\tilde{a}_{k,2^k(2i+1)+j}
	\end{pmatrix}
	=
	\begin{pmatrix}
		1 &  \omega^{[2i]_p} \\
		1 & -\omega^{[2i]_p}
	\end{pmatrix}
	\begin{pmatrix}
		\tilde{a}_{k+1,2^k(2i  )+j} \\
		\tilde{a}_{k+1,2^k(2i+1)+j}
	\end{pmatrix}
\end{equation*}
for $k\in\{0,\dotsc,p-1\}$, $i\in\{0,\dotsc,2^{p-k-1}-1\}$ and $j\in\{0,\dotsc,2^k-1\}$. Thus, it is straight-forward to derive a version of Algorithm~\ref{alg:tft} that works by computing the values $\tilde{a}_{k,i}$ in-place of the values $a_{k,i}$. Reversing the steps of this new algorithm and inverting its transformations then yields Algorithm~\ref{alg:itft}.

\textfloatsep = 14.0pt plus 2.0pt minus 4.0pt

\begin{algorithm}
	\caption{
		$\textproc{ITFT}(
			\ell,
			\omega^{n/2^{\ceil{\log_2\ell}}},
			2^{-1},
			(\x{0},\dotsc,\x{\ell-1}))
		$}
	\label{alg:itft}
	\begin{algorithmic}[1]
		\Require the transform length $\ell\in\{1,\dotsc,n\}$, precomputed elements $\omega^{n/2^{\ceil{\log_2\ell}}}$ and $2^{-1}$, and an array $(\x{0},\dotsc,\x{\ell-1})$ containing $\TFT_{\omega,\ell}(\vec{a})$ for some vector $\vec{a}\in\A^\ell$.
		\Ensure the vector $\vec{a}$ written to $(\x{0},\dotsc,\x{\ell-1})$.
		\State $m\set\ceil{\log_2\ell}$, $v\set\ord_2\ell$
		\For{$k=0,1,\dotsc,m-2$}
			\State
				$q\set\floor{\ell/2^{k+1}}$
				\For{$j=0,\dotsc,2^k-1$}
					\State
						$\begin{pmatrix}
							\x{j}     \\
							\x{2^k+j}
						\end{pmatrix}
						\set
						\begin{pmatrix}
							1 &  1 \\
							1 & -1
						\end{pmatrix}
						\begin{pmatrix}
							\x{j}     \\
							\x{2^k+j}
						\end{pmatrix}$
				\EndFor
			\For{$(i,\alpha)\in\{(i,\omega^{-[2i]_p})\mid i\in\{1,\dotsc,q-1\}\}$}
				\For{$j=0,\dotsc,2^k-1$}
					\State
						$\begin{pmatrix}
							\x{2^k(2i  )+j} \\
							\x{2^k(2i+1)+j}
						\end{pmatrix}
						\set
						\begin{pmatrix}
							1      & 1       \\
							\alpha & -\alpha
						\end{pmatrix}
						\begin{pmatrix}
							\x{2^k(2i  )+j} \\
							\x{2^k(2i+1)+j}
						\end{pmatrix}$
				\EndFor
			\EndFor
		\EndFor
		\For{$k=m-2,m-3,\dotsc,v+1$}
			\State
				$q\set\floor{\ell/2^{k+1}}$,
				$r\set\ell-2^{k+1}q$,
				$q'\set q-2^{m-k-2}$,
				$\alpha\set\omega^{[2q]_p}$
			\If{$r>2^k$}
				\For{$j=r-2^k,\dotsc,2^k-1$}
					\State
						$\x{2^k(2q'+1)+j}
						\set
						\x{2^k(2q)+j}-\alpha\x{2^k(2q'+1)+j}$
				\EndFor
			\Else
				\For{$j=r,\dotsc,2^k-1$}
					\State
						$\x{2^k(2q')+j}
						\set
						2^{-1}(\x{2^k(2q')+j}+\alpha\x{2^k(2q'+1)+j})$
				\EndFor
			\EndIf
		\EndFor
		\For{$k=v,v+1,\dotsc,m-2$}
			\State
				$q\set\floor{\ell/2^{k+1}}$,
				$r\set\ell-2^{k+1}q$,
				$q'\set q-2^{m-k-2}$
			\If{$r>2^k$}
				\State $\alpha\set\omega^{-[2q]_p}$
				\For{$j=0,\dotsc,r-2^k-1$}
					\State
						$\begin{pmatrix}
							\x{2^k(2q)  +j} \\
							\x{2^k(2q+1)+j}
						\end{pmatrix}
						\set
						\begin{pmatrix}
							1      & 1       \\
							\alpha & -\alpha
						\end{pmatrix}
						\begin{pmatrix}
							\x{2^k(2q)  +j} \\
							\x{2^k(2q+1)+j}
						\end{pmatrix}$
				\EndFor
				\For{$j=r-2^k,\dotsc,2^k-1$}
					\State
						$\begin{pmatrix}
							\x{2^k(2q)  +j} \\
							\x{2^k(2q'+1)+j}
						\end{pmatrix}
						\set
						\begin{pmatrix}
							1      &  1      \\
							\alpha & -\alpha
						\end{pmatrix}
						\begin{pmatrix}
							\x{2^k(2q)  +j} \\
							\x{2^k(2q'+1)+j}
						\end{pmatrix}$
				\EndFor
			\Else
				\State $\alpha\set\omega^{[2q]_p}$
				\For{$j=0,\dotsc,r-1$}
					\State $\x{2^k(2q)+j}\set 2\x{2^k(2q)+j}-\alpha\x{2^k(2q'+1)+j}$
				\EndFor
				\For{$j=r,\dotsc,2^k-1$}
				\State $\x{2^k(2q')+j}\set 2\x{2^k(2q')+j}-\alpha\x{2^k(2q'+1)+j}$
				\EndFor\label{itft:push_down_left_only_end}
			\EndIf
		\EndFor
		\State $\alpha\set(2^{-1})^{m-1}$
		\For{$j=\ell-2^{m-1},\dotsc,2^{m-1}-1$}
			\State $\x{j}\set\alpha\x{j}$
		\EndFor
		\State $\alpha\set 2^{-1}\alpha$
		\For{$j=0,\dotsc,\ell-2^{m-1}-1$}
			\State
				$\begin{pmatrix}
					\x{j}         \\
					\x{2^{m-1}+j}
				\end{pmatrix}
				\set
				\alpha
				\begin{pmatrix}
					1 &  1 \\
					1 & -1
				\end{pmatrix}
				\begin{pmatrix}
					\x{j}         \\
					\x{2^{m-1}+j}
				\end{pmatrix}$	
		\EndFor
	\end{algorithmic}
\end{algorithm}

The twiddle factors in Algorithm~\ref{alg:itft} can be computed by the methods of Section~\ref{sec:twiddle_factors}. Consequently, by assuming that all multiplications by $2$ in the algorithm are implemented as additions, a proof similar to that of Lemma~\ref{lem:tft_complexity} provides the following bounds.

\begin{lemma}\label{lem:itft_complexity} Algorithm~\ref{alg:itft} performs at most $(\ell/2)\floor{\log_2\ell}+2\ell+\bigO(\log^2\ell)$ multiplications by powers of $\omega$, at most $2^{\ceil{\log_2\ell}}+\bigO(\log\log\ell)$ multiplications by powers of~$2^{-1}$, and at most $\ell\floor{\log_2\ell}+3\ell$ additions or subtractions.
\end{lemma}

The assumption that $2$ is a unit in $\A$ may be replaced by the weaker assumption that $2$ is not a zero-divisor in $\A$ if all multiplications by powers of $2^{-1}$ in Algorithm~\ref{alg:itft} are replaced by \emph{in-place} evaluations of the corresponding maps $\Div_{2^k}:2^k\A\rightarrow\A$ defined by $2^ka\mapsto a$ for $a\in\A$.

\bibliographystyle{amsplain}

\end{document}